%% file: JournalVersionLearningMDP.tex
\documentclass[10pt,conference,twocolumn]{IEEEtran}
\usepackage{amsmath,amssymb,graphicx,epsfig,amsthm,citesort}

\newtheorem{lemma}{Lemma}
\newtheorem{definition}{Definition}
\newtheorem{proposition}{Proposition}
\usepackage{epstopdf}
\usepackage{balance}
\usepackage{multirow}
\usepackage{mathtools}
\usepackage{algorithm}
\usepackage{algpseudocode}
\usepackage{tabu}
\usepackage{subfig}

\allowdisplaybreaks

\IEEEoverridecommandlockouts

\newcommand{\pds}[1]{\ensuremath{\widetilde{#1}}}

\title{\huge Structural Properties of Optimal Transmission Policies\\ for Delay-Sensitive Energy Harvesting Wireless Sensors\vspace{-0.2cm}}

\author{\IEEEauthorblockN{Nikhilesh Sharma\IEEEauthorrefmark{1}, Nicholas Mastronarde\IEEEauthorrefmark{1}\thanks{The work of N. Mastronarde and J. Chakareski was supported in part by the NSF under awards ECCS-1711335 and ECCS-1711592, respectively.}, and Jacob Chakareski\IEEEauthorrefmark{2}\vspace{-0.3cm}} \\
\IEEEauthorblockA{\IEEEauthorrefmark{1}Dept. Electrical Engineering, U. Buffalo, \IEEEauthorrefmark{2}Dept. Electrical and Computer Engineering, U. Alabama\vspace{-0.4cm}}}

\begin{document}
\maketitle

\begin{abstract}
We consider an energy harvesting sensor transmitting latency-sensitive data over a fading channel. We aim to find the optimal transmission scheduling policy that minimizes the packet queuing delay given the available harvested energy. We formulate the problem as a Markov decision process (MDP) over a state-space spanned by the transmitter's buffer, battery, and channel states, and analyze the structural properties of the resulting optimal value function, which quantifies the long-run performance of the optimal scheduling policy. We show that the optimal value function (i) is non-decreasing and has increasing differences in the queue backlog; (ii) is non-increasing and has increasing differences in the battery state; and (iii) is submodular in the buffer and battery states. Our numerical results confirm these properties and demonstrate that the optimal scheduling policy outperforms a so-called greedy policy in terms of sensor outages, buffer overflows, energy efficiency, and queuing delay.
\end{abstract}


\input{intro}
\input{system}
\input{formulation}
\input{structure}
\input{sim}
\input{conclusion}
\input{proofs}

\balance
\bibliographystyle{IEEEtran}
\bibliography{refs}

\end{document}

%% file: intro.tex
\section{Introduction}
\label{sec:intro}
Energy-constrained wireless sensors are increasingly used for latency-sensitive applications such as real-time remote visual sensing~\cite{Chakareski:15}, Internet of Things (IoT), body sensor networks~\cite{seyedi2010energy}, smart grid monitoring, and cyber-physical systems. However, these sensors are subject to time-varying channel conditions and generate stochastic traffic loads,
which makes it very challenging for them to provide the necessary Quality of Service (QoS) to support latency-sensitive applications.
This is further complicated by the introduction of wireless sensors powered by energy harvested from the environment (e.g., ambient light or RF energy~\cite{vullers2010energy}). Although energy harvesting  sensors (EHSs) can operate autonomously in (possibly remote) areas without access to power lines and without the need to change their batteries, the stochastic nature of harvested energy sources poses new challenges in sensor power management, transmission power allocation, and transmission scheduling.

An important body of work focuses on offline computation of optimal transmission policies for EHSs~\cite{gurakanenergy,lu2014dynamic,sharma2010optimal,gunduz2014designing}. In particular, \cite{gurakanenergy} considers a multi-access channel with two EHSs and derives the optimal offline transmission power and rate allocations that maximize the sum rate, given a priori known energy and traffic arrival processes. 
\cite{gunduz2014designing} identifies Markov decision processes (MDPs~\cite{puterman2014markov}) as a useful tool for optimizing EHSs in unpredictable environments with only causal information about the past and present. \cite{sharma2010optimal} formulates both throughput-optimal and delay-optimal energy management policies as MDPs.
Though these studies identify numerous techniques for calculating optimal policies, they do not provide general insights into their structures.

Another body of work focuses on characterizing the structure of optimal transmission policies for EHSs~\cite{seyedi2010energy,ozel2011transmission,ho2012optimal,yang2012optimal1,yang2012optimal,aprem2013transmit}. 
Numerous studies have shown that optimal power allocation policies for EHSs have various water-filling structures~\cite{ozel2011transmission,ho2012optimal,yang2012optimal1}.
Other types of structural results are derived in~\cite{yang2012optimal,aprem2013transmit}. 
In particular, \cite{yang2012optimal} assumes that known amounts of data and  energy arrive over a finite time horizon, and aims to minimize the total amount of time to transmit all data. They show that the optimal policy uses transmission rates that increase over time. 
\cite{aprem2013transmit} formulates outage-optimal power control policies for EHSs, showing that the optimal policy for the underlying MDP is threshold in the battery state for the special case of binary transmission power levels. 

We study an EHS transmitting delay-sensitive data over a fading channel. We assume that it uses a fixed transmission power, can transmit at most one packet in each time slot, and experiences a variable packet loss rate depending on the channel conditions. Under these assumptions, we aim to understand the structure of optimal transmission scheduling policies that minimize the packet queuing delay given the available harvested energy. Our contributions are as follows:
\begin{itemize}
\item We formulate the delay-sensitive energy harvesting scheduling (DSEHS) problem as an MDP that takes into account the stochastic traffic load, harvested energy, and channel conditions experienced by the EHS.
\item We show that the optimal value function, which quantifies the long-run performance of the optimal scheduling policy, (i) is non-decreasing and has increasing differences in the queue backlog; (ii) is non-increasing and has increasing differences in the battery state; and (iii) is submodular in the buffer and battery states.
\item  Our numerical results confirm these properties and demonstrate that the optimal scheduling policy outperforms a so-called greedy policy in terms of sensor outage, buffer overflow,  energy efficiency, and queuing delay.
\end{itemize}

Our advances can facilitate \textit{online learning} of optimal policies at lower complexity and enable efficient self-organizing operation of next generation IoT sensing systems (see, e.g.,~\cite{toorchi2016fast}). Such studies fall outside the scope of our paper.

The remainder of this paper is organized as follows. We introduce the system model in Section~\ref{sec:sys}, formulate the DSEHS problem in Section~\ref{sec:formulation}, analyze the structural properties of the DSEHS problem in Section~\ref{sec:struct-results}, present our numerical results in Section~\ref{sec:sim}, and conclude in Section~\ref{sec:con}.

%% file: system.tex
\section{Wireless Sensor Model}
\label{sec:sys}

We consider a time-slotted single-input single-output (SISO) point-to-point wireless communication system in which an energy harvesting sensor transmits latency-sensitive imagery data over a fading channel. The system model is depicted in Fig. 1. The system comprises two buffers: a packet buffer with (possibly infinite) size $N_b$ and an energy buffer (battery) with finite size $N_e$. We assume that time is divided into slots with length $\Delta T$ (seconds) and that the system's state in the $n$th time slot is denoted by $s^{n}\triangleq(b^{n},e^{n},h^{n})\in \mathcal{S}$, where $b^{n}\in \mathcal{S}_b=\left\{ 0,1,...,N_b\right\}$ is the packet buffer state (i.e., the number of backlogged data packets), $e^{n}\in \mathcal{S}_e=\left\{ 0,1,...,N_e\right\}$ is the battery state (i.e., the number of energy packets in the battery), and $h^{n} \in \mathcal{S}_h$ is the channel fading state. At the start of the $n$th time slot, the transmission scheduler observes the state of the system and takes the binary scheduling action $a^{n} \in \mathcal{A} = \{0,1\}$, where $a^{n} = 1$ indicates that it transmits the head-of-line packet in the queue and $a^{n} = 0$ otherwise.

\subsection{Channel model}
We assume a block-fading channel, meaning that the channel is constant during each time slot and may change from one slot to the next. Similar to prior work \cite{zordandesign,sharma2010optimal,ozel2011transmission,mastronarde2011fast,djonin2007mimo}, we assume that the channel state $h^{n}\in \mathcal{S}_h$ is known to the transmitter at the start of each time slot, that $\mathcal{S}_h$ denotes a finite set of $N_h$ channel states, and that the evolution of the channel state can be modeled as a finite state Markov chain with transition probability function $P^{h}(h^\prime|h)$.

\begin{figure}
\centering
  \includegraphics[width=0.9\linewidth]{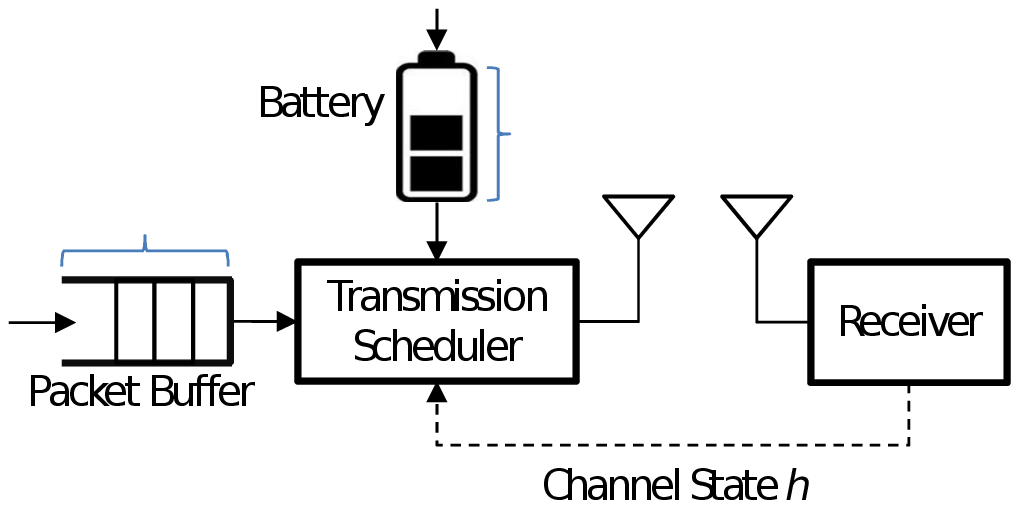}
  \caption{System block diagram.}
  \label{fig:1}
  \vspace{-0.5cm}
\end{figure}

\subsection{Energy harvesting model}
Similar to prior work \cite{zordandesign,aprem2013transmit}, we assume that battery energy is stored in the form of energy packets.
Let $e_{H}^{n} \in \mathcal{E} =\{0,1,\ldots,M_e\}$ denote the number of energy packets that are available for harvesting in the $n$th time slot and let $P^{e_H}(e_{H})$ denote the energy packet arrival distribution. Energy packets that arrive in time slot $n$ can be used in future time slots. Therefore, the battery state at the start of time slot $n+1$ can be found through the following recursion:
\begin{equation}\label{eq:e-recursion}
e^{n+1}=\min(e^{n}-e_{\text{TX}}(a^n)+e_{H}^{n},N_e),
\end{equation}
\noindent
where $e_{\text{TX}}(a^n)$ denotes the number of energy packets consumed in time slot $n$ given the scheduling action $a^n$. We assume that the wireless sensor uses a fixed transmission power $P_\text{TX}$ (energy packets per second); therefore,
\begin{equation}
	e_{\text{TX}}(a^n) = a^n P_\text{TX} \Delta T  = a^n e_{\text{TX}}\mbox{ (energy packets)}.
\end{equation}
For simplicity, we assume that the transmission energy $e_{\text{TX}}$ is an integer multiple of energy packets. Note that the transmission action $a^n$ in time slot $n$ cannot use more energy than is available in the battery, i.e., $a^{n} e_{\text{TX}} \leq e^{n}$.

Given the current state $s = (b,e,h)$ and action $a$, the probability of observing battery state $e^\prime$ in the next slot is:
\begin{equation}\label{eq:battery-tpf}
	P^{e}(e^\prime|e,a) = \mathbb{E}_{e_H}[\mathbb{I}_{\{e^\prime = \min(e-a \cdot e_{\text{TX}}+e_{H},N_e)\}}],
\end{equation}
where $\mathbb{I}_{\{\cdot\}}$ is an indicator variable that is set to 1 when ${\{\cdot\}}$ is true and is set to 0 otherwise.



\subsection{Traffic model}
Let $l^{n} \in \mathcal{L} = \{0,1,\ldots,M_l\}$ denote the number of data packets generated by the sensor in the $n$th time slot and let $P^l(l)$ denote the data packet arrival distribution. The buffer state in slot $n+1$ can be found through the following recursion:
\begin{equation}\label{eq:b-recursion}
b^{n+1}=\min(b^{n}-f^{n}(a^n,h^n)+l^{n},N_b),
\end{equation}
\noindent
where $f^{n}(a^n,h^n)$ is the number of packets transmitted successfully in time slot $n$ and $f^{n}(a^n,h^n) \leq a^{n} \leq b^n$. Note that new packet arrivals, and packets that are not successfully received, must be (re)transmitted in a future time slot. Assuming independent and identically distributed (i.i.d.) bit errors, we can characterize $f^n$ as a binary random variable with conditional probability mass function
\begin{equation}\label{eq:goodput-pmf}
	P^f(f|a,h) =
    \begin{cases}
    	1, & \mbox{if $f = 0$ and $a = 0$}, \\
        0, & \mbox{if $f = 1$ and $a = 0$}, \\
        q(h), & \mbox{if $f = 0$ and $a = 1$}, \\
        1-q(h), & \mbox{if $f = 1$ and $a = 1$},
    \end{cases}
\end{equation}
where $q(h)$ is the packet loss rate (PLR) in channel state $h$. Since the transmission power is fixed, $q(h^+) < q(h^-)$ if $h^+ > h^-$.
We will refer to $P^f(f|a,h)$ as the \textit{goodput distribution}. 

Given the current state $s = (b,e,h)$ and action $a$, the probability of observing buffer state $b^\prime$ in the next time slot is:
\begin{equation}\label{eq:buffer-tpf}
	P^{b}(b^\prime|[b,h],a) = \mathbb{E}_{f,l}[\mathbb{I}_{\{b^\prime = \min(b-f+l,N_b)\}}]. 
\end{equation}
%

%% file: formulation.tex
\section{The Delay-Sensitive Energy-Harvesting Scheduling (DSEHS) Problem}
\label{sec:formulation}

Let $\pi : \mathcal{S} \rightarrow \mathcal{A}$ denote a \textit{policy} that maps states to actions. The objective of the DSEHS problem is to determine the optimal policy $\pi^*$ that minimizes the average packet queuing delay given the available energy. However, this does not mean that the policy should greedily transmit packets whenever there is enough energy to do so. Instead, it may be beneficial to abstain from transmitting packets in bad channel states and wait to transmit them in good channel states to reduce retransmissions and conserve scarce harvested energy. On the other hand, the policy should not be too conservative. Instead, if the battery is (nearly) full, transmitting a packet will make room for more harvested energy, which otherwise would be lost due to the finite battery size.
To balance these considerations, we formulate the DSEHS problem as an MDP~\cite{puterman2014markov}.

We define a \textit{buffer cost} to penalize large queue backlogs. Formally, we define the buffer cost as the sum of the \textit{holding cost} and the expected \textit{overflow cost} with respect to the arrival and goodput distributions, i.e., 
\begin{equation}\label{eq:cost}
  c([b,h],a) = b + \mathbb{E}_{f,l}[\{\eta\max(b-f+l-N_b,0)\}],
\end{equation}
In~\eqref{eq:cost}, the holding cost is equal to the buffer backlog, which is proportional to the queuing delay by Little's theorem~\cite{bertsekas1987data}. The overflow cost imposes a penalty $\eta$ for each dropped packet.

Formally, the DSEHS problem's objective is to determine the scheduling policy that solves the following optimization:
\begin{equation}
\begin{aligned}
& \underset{\pi \in \Pi}{\text{minimize}}
& &  \mathbb{E}\left[\sum\nolimits_{n=0}^{\infty}(\gamma)^{n}c(s^{n},\pi(s^{n}))\right],\\
\end{aligned}\label{eq:discounted_cost}
\end{equation}
where $\gamma \in [0, 1)$ is the discount factor, $\Pi$ is the set of all possible policies, and the expectation is taken over the sequence of states, which are governed by a controlled Markov chain with transition probabilities:
\begin{align}
P(s^\prime|s,a)=P^{b}(b^\prime|[b,h],a)  P^{h}(h^\prime|h)  P^{e}(e^\prime|e,a).\label{eq:tpf}
\end{align}
The optimal solution to~\eqref{eq:discounted_cost} satisfies the following Bellman equation, $\forall s \in \mathcal{S}$:
\begin{align}
\lefteqn{V^{*}(s)} \nonumber \\
&= \min_{a \in \mathcal{A}(s)}  \biggl\{c(s, a)+\gamma\sum_{s^\prime \in \mathcal{S}}P(s^\prime | s, a)V^{*}(s^\prime)\biggr\}, \nonumber \\
&= \min_{a\in \mathcal{A}(b, e)}  \biggl\{ c([b, h], a) + \gamma \mathbb{E}_{l,f,e_H,h'} \nonumber \\
& \quad  [V^{*}(\min(b - f + l, N_b), \min(e - a \cdot e_\text{TX} + e_H, N_e), h^\prime)] \biggr\} \nonumber \\
&\triangleq \min_{a \in \mathcal{A}(s)}  Q^{*}(s, a),  \label{eq:value}
\end{align}

\noindent where $\mathcal{A}(b, e)$ is the set of feasible actions given the buffer and battery states (i.e., $\mathcal{A}(b, e) = \{0, 1\}$ if $b > 0$ and $e \geq e_{TX}$, and is $\{0\}$ otherwise), $V^{*}(s)$ is the optimal \textit{state-value function}, and $Q^{*}(s,a)$ is the optimal \textit{action-value function}. The optimal policy $\pi^{*}(s)$ can be determined by taking the action in each state that minimizes the r.h.s. of~\eqref{eq:value}.

\subsection{Post-Decision State Based Dynamic Programming}\label{sec:pds-dp}

We will find it useful throughout our analysis to work with so-called post-decision states (PDSs) rather than conventional states. A PDS, $\pds{s} \triangleq (\pds{b}, \pds{e}, \pds{h}) \in \mathcal{S}$, denotes a state of the system after all known dynamics have occurred, but before the unknown dynamics occur~\cite{mastronarde2011fast}. In the DSEHS problem, 
\begin{equation}\label{eq:pds}
\tilde{s}^{n}=(\pds{b}^n,\pds{e}^n,\pds{h}^n)=([b^{n}-f^{n}],[e^{n}-a^{n} \cdot e_{\text{TX}}],h^{n})
\end{equation}
is the PDS in time slot $n$. 
The buffer's PDS $\pds{b}^n=b^{n}-f^{n}$ characterizes the buffer state after a packet is transmitted (if any), but
before any new packets arrive; the battery's PDS $\pds{e}^n=e^{n}-a^{n} \cdot e_{\text{TX}}$ characterizes the battery state after an energy packet is consumed (if any), but before any new energy packets arrive; and the channel's PDS $\pds{h}^{n} = h^{n}$ is the same as the channel state at time $n$. In other words, the PDS incorporates all of the known information about the transition from state $s^n$ to state $s^{n+1}$ after taking action $a^n$. Meanwhile, the unknown dynamics in the transition from state $s^{n}$ to $s^{n+1}$, i.e., the channel state transition from $h^n$ to $h^{n+1} \sim P^h(\cdot|h^{n})$, the data packet arrivals $l^{n} \sim P^l(\cdot)$, and the energy packet arrivals $e_{H}^{n} \sim P^{e_H}(\cdot)$ are not included in the PDS. Importantly, the next state can be expressed in terms of the PDS as follows:
\begin{align}
s^{n + 1} &= (b^{n + 1}, e^{n + 1}, h^{n + 1}) \nonumber \\
&= (\min(\pds{b}^n + l^n, N_b), \min(\pds{e}^n + e_H^n, N_e), h^{n + 1}).
\end{align}

Just as we defined a value function over the conventional states, we can define a PDS value function over the PDSs. Let $\pds{V}^{*}$ denote the optimal PDS value function. $\pds{V}^{*}$ and $V^{*}$ are related by the following Bellman equations:
\begin{multline}\label{eq:V_to_PDSV}
\pds{V}^{*}(\pds{s}) = \eta \mathbb{E}_l [\max(\pds{b} + l - N_b, 0)] + \\
\gamma \mathbb{E}_{l,e_H,h'}[V^{*}(\min(\pds{b} + l, N_b), \min(\pds{e} + e_H, N_e), h^\prime)] 
\end{multline}
\begin{equation}\label{eq:PDSV_to_V}
V^{*}(s) = 
\min_{a \in \mathcal{A}(b, e)} \left\{b + \mathbb{E}_f [\pds{V}^{*}(b - f, e - a \cdot e_{TX}, h)]\right\}
%
\end{equation}
Knowing $\pds{V}^{*}(\pds{s})$, $\pi^{*}(s)$ can be found by taking the action in each state that minimizes the r.h.s. of \eqref{eq:PDSV_to_V}.

Algorithm~\ref{alg:pds-value-iter} presents a value iteration algorithm for computing the PDS value function \textit{offline}. Although it is too complex to be implemented on an EHS, its iterative structure facilitates the use of mathematical induction to derive structural properties of the optimal PDS value function $\pds{V}^{*}(\pds{s})$ (see Section~\ref{sec:struct-results}). 


\begin{algorithm}
\caption{Post-Decision State Value Iteration}
\label{alg:pds-value-iter}
\begin{algorithmic}[1]
\State \textbf{initialize} $\pds{V}_0(\pds{b}, \pds{e}, \pds{h}) = 0$ for all $(\pds{b}, \pds{e}, \pds{h}) \in \mathcal{S}$ and $\tau = 0$
\Repeat
\State $\Delta \leftarrow 0$
\For {$(b, e, h) \in \mathcal{S}$}
\State Update the value function:
\begin{multline}\label{eq:update-v}
V_{\tau}(b,e,h) \leftarrow \\
\min_{a \in \mathcal{A}(b, e)} \biggl\{b + \mathbb{E}_f [\pds{V}_{\tau}(b - f, e - a \cdot e_{TX}, h)]\biggr\}
\end{multline}
\EndFor
\For {$(\pds{b}, \pds{e}, \pds{h}) \in \mathcal{S}$}
\State Update the PDS value function:
\begin{multline}
\pds{V}_{\tau + 1}(\pds{b}, \pds{e}, \pds{h}) \leftarrow \eta \mathbb{E}_l [\max(\pds{b} + l - N_b, 0)] + \\ 
\gamma \mathbb{E}_{l,e_H,h'} [V_{\tau}(\min(\pds{b} + l, N_b), \min(\pds{e} + e_H, N_e), h^\prime)]
\end{multline}
\State $\Delta \leftarrow \max(\Delta, |\pds{V}_{\tau}(\pds{b}, \pds{e}, \pds{h}) - \pds{V}_{\tau + 1}(\pds{b}, \pds{e}, \pds{h})|)$
\EndFor
\State $\tau \leftarrow \tau + 1$
\Until {$\Delta < \theta$ (a small positive constant)}
\end{algorithmic}
\end{algorithm}

%% file: structure.tex
\section{Structural Properties}\label{sec:struct-results}
In this section, we analyze the structural properties of the optimal PDS value function $\pds{V}^{*}(s)$. Understanding
such properties is important because: 
(i) they provide insights into the optimization problem and the system being optimized; 
(ii) they reveal ways in which the solution can be represented compactly, with limited memory; and 
(iii) they can facilitate efficient \textit{online} computation of the optimal policy using reinforcement learning (see, e.g.,~\cite{toorchi2016fast,mastronarde2011fast}).
In this paper, we focus on point (i) above.
We begin by introducing three important definitions and providing an overview of our results. Then, in Section~\ref{sec:properties-cost-tpf}, we analyze the properties of the cost and transition probability functions and, in Section~\ref{sec:V-structure}, we analyze several key properties of the conventional value function. These properties are all needed to prove our main results, which are presented in Section~\ref{sec:PDSV-structure}.


The first useful definition is that of integer convexity.

\begin{definition}
(Integer Convex~\cite{djonin2007mimo}): An integer convex function $f(n): \mathcal{N} \rightarrow \mathbb{R}$ on a set of integers $\mathcal{N} \in \{0, 1, \ldots ,N\}$ is a function that has increasing differences in $n$, i.e., 
\begin{equation}
	f(n_1 + m) - f(n_1) \leq f(n_2 + m) - f(n_2)
\end{equation}
for $n_1 < n_2$, $n_1, n_2, n_1 + m, n_2 + m \in \mathcal{N}$. 
\end{definition} 

Our main results establish that the PDS value function has increasing differences in the PDS buffer state $\pds{b}$ (Proposition~\ref{prop:incr-diff-PDSV-b}) and the PDS battery state $\pds{e}$ (Proposition~\ref{prop:incr-diff-PDSV-e}).

The second useful definition is that of stochastic dominance.

\begin{definition}
(Stochastic Dominance~\cite{djonin2007mimo}): Let $\theta(x)$ be a random variable parameterized by some $x \in \mathbb{R}$. If $P(\theta(x_1) \geq a) \geq P(\theta(x_2) \geq a)$ for all $x_1 \geq x_2$ and for all $a \in \mathbb{R}$, then we say that $\theta(x)$ is first-order stochastically increasing in $x$. If $\theta(x)$ is first-order stochastically increasing in $x$, then
\begin{equation}\label{eq:stochastic_increasing}
\mathbb{E}[u(\theta(x_1))] \geq \mathbb{E}[u(\theta(x_2))] 
\end{equation}
for all non-decreasing functions $u(x)$. The reverse inequality holds for all non-increasing functions $u(x)$.
\end{definition}

In Section~\ref{sec:properties-cost-tpf}, we establish that the buffer and battery state transition probabilities defined in~\eqref{eq:buffer-tpf} and~\eqref{eq:battery-tpf}, respectively, are first-order stochastically increasing in the buffer state $b$ and the battery state $e$, respectively. In Section~\ref{sec:V-structure}, we use these properties -- combined with~\eqref{eq:stochastic_increasing} -- to show that the value function is non-decreasing in the buffer state $b$ and is non-increasing in the battery state $e$. These results help us establish integer convexity of the PDS value function in Section~\ref{sec:PDSV-structure}.

Lastly, we define the concept of a submodular function.

\begin{definition}
(Submodular \cite{puterman2014markov}): A submodular function $f(x,y):\mathcal{X} \times \mathcal{Y} \rightarrow \mathbb{R}$ on sets of integers $\mathcal{X} \in \{0, 1, \ldots, X\}$ and $\mathcal{Y} \in \{0, 1, \ldots, Y\}$ is a function that has decreasing differences in $(x, y)$, i.e., for $x^+ \geq x^-$ and $y^+ \geq y^-$
\begin{equation}\label{eq:submodular}
f(x^+, y^+) - f(x^+, y^-) \leq f(x^-, y^+) - f(x^-, y^-).
\end{equation}
\end{definition}

In Section~\ref{sec:PDSV-structure}, we prove that the PDS value function is submodular in $(\pds{b}, \pds{e})$ (Proposition~\ref{prop:PDSV-submodular-b-e}).

\subsection{Properties of the Cost and Transition Probability Functions}\label{sec:properties-cost-tpf}

We now present key properties of the cost and transition probability functions that we will need for our main results.
Recall that the cost does not directly depend on the battery state $e$, but that the action $a$ is constrained to be in the set $\mathcal{A}(b,e)$ (i.e., $a$ is constrained to be 0 if $e < e_{TX}$ or $b = 0$). To show this explicitly, we define an auxiliary cost function
\begin{equation}\label{eq:auxiliary-cost}
  d([b,e,h],a) = 
  \begin{cases}
    c([b,h],a), & \mbox{if $b > 0$ and $e \geq e_{TX}$} \\
    c([b,h],0), & \mbox{otherwise,}
  \end{cases}
\end{equation}
where $c([b,h],a)$ is defined in~\eqref{eq:cost}. We omit the proofs of the following three lemmas due to space limitations.

\begin{lemma}\label{lem:cost}
The auxiliary cost $d([b, e, h], a)$ satisfies the following properties:
\begin{enumerate}
\item \label{item:incr-cost-b} The auxiliary cost is non-decreasing in $b$.
\item \label{item:decr-cost-e} The auxiliary cost is non-increasing in $e$.
\end{enumerate}
\end{lemma}
Since the auxiliary cost function satisfies Lemma~\ref{lem:cost}, the cost function $c([b,h],a)$, with $a \in \mathcal{A}(b,e)$, also satisfies it.

\begin{lemma}\label{lem:stoch-dom-e}
The battery state transition probabilities are first-order stochastically increasing in the battery state $e$, i.e.,
\begin{equation*}
\sum_{e' \geq \bar{e} } P^e(e' | e + 1, a) \geq \sum_{e' \geq \bar{e}} P^e(e' | e, a), \quad 0 \leq e < N_e.
\end{equation*}
\end{lemma}
\begin{lemma}\label{lem:stoch-dom-b}
The buffer state transition probabilities are first-order stochastically increasing in the buffer state $b$, i.e.,
\begin{equation*}
\sum_{b' \geq \bar{b} } P^b(b' | [b + 1, h], a) \geq \sum_{b' \geq \bar{b}} P^b(b' | [b, h], a), \quad 0 \leq b < N_b.
\end{equation*}
\end{lemma}

Lemma~\ref{lem:stoch-dom-e} (Lemma~\ref{lem:stoch-dom-b}) implies that the next battery (buffer) state has a higher probability of exceeding a threshold if the current battery (buffer) state is larger.

\subsection{Properties of the Conventional State Value Function}\label{sec:V-structure}\label{sec:V-structure}

\begin{lemma}\label{lem:V-non-decr-b}
The optimal value function $V^*(b, e, h)$ is non-decreasing in the buffer state $b$.
\end{lemma}
\begin{proof}
The proof is given in the appendix.
\end{proof}

\begin{lemma}\label{lem:V-non-incr-e}
The optimal value function $V^*(b, e, h)$ is non-increasing in the battery state $e$.
\end{lemma}

\begin{proof}
We omit the proof as it is similar to Lemma~\ref{lem:V-non-decr-b}. 
\end{proof}

The following lemma is needed for the inductive steps in our main results (propositions~\ref{prop:incr-diff-PDSV-b}, \ref{prop:incr-diff-PDSV-e}, and~\ref{prop:PDSV-submodular-b-e}).

\begin{lemma}\label{lem:PDSV-to-V}
The following properties are propagated from the PDS value function $\pds{V}(\pds{b}, \pds{e}, \pds{h})$ to the conventional value function $V(b, e, h)$ through the Bellman equation given in~\eqref{eq:PDSV_to_V}:
\begin{enumerate}
\item \label{item:incr-diff-b} If $\pds{V}(\pds{b}, \pds{e}, \pds{h})$ has increasing differences $\pds{b}$, then $V(b, e, h)$ has increasing differences in $b$.

\item \label{item:incr-diff-e} If $\pds{V}(\pds{b}, \pds{e}, \pds{h})$ has increasing differences in $\pds{e}$, then $V(b, e, h)$ has increasing differences in $e$.

\item \label{item:submod-b-e} If $\pds{V}(\pds{b}, \pds{e}, \pds{h})$ is submodular in $(\pds{b},\pds{e})$, then $V(b, e, h)$ is submodular in $(b,e)$.

\end{enumerate}
\end{lemma}

\begin{proof}
The proof is given in the appendix.
\end{proof}

Lemma~\ref{lem:PDSV-to-V} implies that the   PDS value function's properties are propagated to the conventional value function during the value function update step in Algorithm~\ref{alg:pds-value-iter} (see~\eqref{eq:update-v}). 

\subsection{Properties of the Post-Decision State Value Function}\label{sec:PDSV-structure}

We now prove that the optimal PDS value function has increasing differences in the buffer's PDS $\pds{b}$ and the battery's PDS $\pds{e}$, and decreasing differences in $(\pds{b}, \pds{e})$ (i.e., it is submodular in $(\pds{b}, \pds{e})$). We then discuss the meaning of these results.

\begin{proposition}\label{prop:incr-diff-PDSV-b}
If the packet buffer has infinite size ($N_b = \infty$), then $\pds{V}^*(\pds{b}, \pds{e}, \pds{h})$ has increasing differences in $\pds{b}$, i.e.,
\begin{multline}\label{eq:incr-diff-PDSV-b}
\pds{V}^*(\pds{b}, \pds{e}, \pds{h}) - \pds{V}^*(\pds{b} - 1, \pds{e}, \pds{h}) \\
\leq \pds{V}^*(\pds{b} + 1, \pds{e}, \pds{h}) - \pds{V}^*(\pds{b}, \pds{e}, \pds{h}).
\end{multline}
\end{proposition}
\begin{proof}
The proof is given in the appendix.
\end{proof}

\begin{proposition}\label{prop:incr-diff-PDSV-e}
$\pds{V}^*(\pds{b},\pds{e},\pds{h})$ has increasing differences in $\pds{e}$, i.e.,
\begin{multline}\label{eq:incr-diff-PDSV-e}
\pds{V}^*(\pds{b}, \pds{e}, \pds{h}) - \pds{V}^*(\pds{b}, \pds{e} - 1, \pds{h}) \\
\leq \pds{V}^*(\pds{b}, \pds{e} + 1, \pds{h}) - \pds{V}^*(\pds{b}, \pds{e}, \pds{h}).
\end{multline}
\end{proposition}
\begin{proof}
We omit the proof as it is similar Proposition~\ref{prop:incr-diff-PDSV-b}.
\end{proof}

\begin{proposition}\label{prop:PDSV-submodular-b-e}
$\pds{V}^*(\pds{b}, \pds{e}, \pds{h})$ is submodular in $(\pds{b}, \pds{e})$, i.e., 
\begin{multline}\label{eq:PDSV-submodular-b-e}
\pds{V}^*(\pds{b}+1, \pds{e}+1, \pds{h}) - \pds{V}^*(\pds{b}, \pds{e}+1, \pds{h}) \\
\leq \pds{V}^*(\pds{b} + 1, \pds{e}, \pds{h}) - \pds{V}^*(\pds{b}, \pds{e}, \pds{h}).
\end{multline}
\end{proposition}

\begin{proof}
The proof is given in the appendix.
\end{proof}

Together, Proposition~\ref{prop:incr-diff-PDSV-b} and Lemma~\ref{lem:V-non-decr-b} imply that the cost to serve an additional data packet increases with the queue backlog. 
Although we were only able to prove that $\pds{V}^*(\pds{b}, \pds{e}, \pds{h})$  has increasing differences in the buffer state for an infinite size buffer, we have not observed any cases in practice where this property does not hold for finite buffers. 
Together, Proposition~\ref{prop:incr-diff-PDSV-e} and Lemma~\ref{lem:V-non-incr-e} imply that the benefit of an additional energy packet decreases with the available battery energy. 
Finally, Proposition~\ref{prop:PDSV-submodular-b-e} implies that data packets and energy packets are \textit{complementary}.
That is, the cost of serving an additional data packet is smaller when more energy is available, and the benefit of having an additional energy packet is greater when more data packets need to be served.

%% file: sim.tex
\section{Numerical Results}
\label{sec:sim}

In Section~\ref{sec:sim-structure}, we illustrate the structural properties of the optimal PDS value function. In Section~\ref{sec:sim-performance}, we compare the optimal scheduling policy against a so-called greedy policy, which always transmits backlogged packets if there is sufficient energy (i.e., $e^n \geq e_{TX}$). The parameters used in our MATLAB-based simulator are given in Table~\ref{tab:simulation-parameters}.

\begin{table*}[h!]
\centering
\caption{Simulation Parameters} 
\label{tab:simulation-parameters}
\begin{tabu}  to  \textwidth { | X[c] | X[c] | X[c] | X[c] | }
\hline
\textbf{Parameter} & \textbf{Value} & \textbf{Parameter} & \textbf{Value} \\
\hline \hline
Packet Buffer Size, $N_b$ & 25 & Transmission Action, $a \in \mathcal{A}$ & $\left\lbrace 0, 1 \right\rbrace$ \\ 
\hline
Energy Buffer Size, $N_e$ & 25 & Transmission Energy, $e_{TX}$ & 1 \\
\hline
Channel States $h \in \mathcal{H}$ & $\left\lbrace 1, 2, ..., 7, 8 \right\rbrace$ & Discount Factor, $\gamma$ & 0.98 \\
\hline
Packet Loss Rate (PLR), $q(h)$ & $\left\lbrace 0.8, 0.7, 0.6, ... , 0.2, 0.1 \right\rbrace$ & Simulation Duration (time slots & 50,000 \\ 
\hline
Packet Arrivals (packets/time slot) & $\left\lbrace 0, 1 \right\rbrace$ & Packet Arrival PMF, $P^l(l)$ & Bernoulli$(p)$ with variable $p$  \\
\hline
Energy Arrivals (packets/time slot) & $\left\lbrace 0, 1 \right\rbrace$ &  Energy Arrival PMF, $P^{e_H}(e_H)$ & Bernoulli$(p)$ with $p = 0.7$  \\
\hline
Overflow Penalty, $\eta$ & 50 & Steady-State Channel Probabilities & (0.071, 0.143, 0.143, 0.143, 0.143, 0.143, 0.143, 0.071) \\
\hline
\end{tabu}
\end{table*}


\subsection{Structural Properties}
\label{sec:sim-structure}
In this section, we assume that the packet and energy arrivals are Bernoulli random variables with parameters $0.4$ and $0.7$, respectively. 
Fig.~\ref{fig:optimal-v} and Fig.~\ref{fig:optimal-policy} show the optimal PDS value function and policy, respectively, under these assumptions. From Fig.~\ref{fig:optimal-v}, it is clear that the optimal PDS value function (i) is non-decreasing and has increasing differences in the queue backlog (Lemma~\ref{lem:V-non-decr-b} and Proposition~\ref{prop:incr-diff-PDSV-b}) and (ii) is non-increasing and has increasing differences in the battery state (Lemma~\ref{lem:V-non-incr-e} and Proposition~\ref{prop:incr-diff-PDSV-e}). 
Fig.~\ref{fig:submodular} shows that $\pds{V}(\pds{b}+1,\pds{e},\pds{h}) - \pds{V}(\pds{b},\pds{e},\pds{h})$ is non-increasing in the battery state \pds{e}, i.e., the optimal PDS value function is submodular in $(\pds{b},\pds{e})$ (Proposition~\ref{prop:PDSV-submodular-b-e}). From Fig.~\ref{fig:optimal-policy}, we observe that the optimal policy is more conservative than the greedy policy because it does not transmit at low battery states.

  

\begin{figure} [h]
\centering
  \subfloat[{Optimal PDS value function}]
  {
  \includegraphics[clip, trim = 1cm 9cm 1cm 9cm, width=0.9\linewidth]{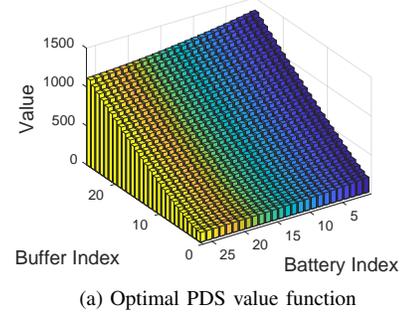}
  \label{fig:optimal-v}
  }
  
  \subfloat[{Optimal policy}]
  {
  \includegraphics[clip, trim = 1cm 9cm 1cm 9cm, width=0.9\linewidth]{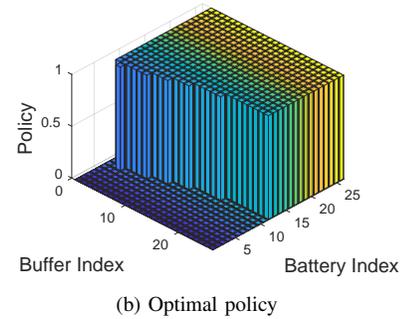}
  \label{fig:optimal-policy}
  }
  \caption{Optimal PDS value function and policy in channel state $h$ with PLR $q(h) = 0.8$.}
  \label{fig:optimal-v-and-policy}
  \vspace{-0.5cm}
\end{figure}

\begin{figure} [h]
\centering
   \includegraphics[clip, trim = 1cm 9cm 1cm 9cm, width=0.9\linewidth]{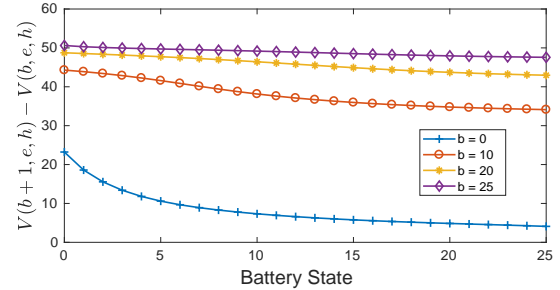}
\   \caption{Submodularity of the PDS value function in $(b, e)$ with PLR $q(h) = 0.8$.}
   \label{fig:submodular}
   \vspace{-0.5cm}
\end{figure}

\subsection{Performance Evaluation}
\label{sec:sim-performance}
We now compare the performance of the optimal and greedy policies assuming that $P^l(l) = \text{Bernoulli}(p)$, where $p \in \{ 0.1, 0.122, 0.144, \ldots, 0.6 \}$,  $P^{e_H}(e_H) = \text{Bernoulli}(0.7)$, and $q(h)=0.8$. Note that the optimal policies were computed offline using Algorithm~\ref{alg:pds-value-iter} and then stored in a lookup table. In Fig.~\ref{fig:avg-buffer-vs-energy-optimal}, we show how the average queue backlog (left axis) and average battery state (right axis) vary with respect to the packet arrival rate. Each measurement is taken from a 50,000 time slot simulation of the corresponding policy. The optimal policy achieves 2.6\% -- 37.4\% lower queue backlogs (19.1\% lower when averaged across all data points) and maintains 0.1\% -- 258.5\% more battery energy (71.1\% higher when averaged across all data points) than the greedy policy.

In Fig.~\ref{fig:avg-outage-vs-avg-overflow}, we show how the buffer overflow (left axis) and battery outage (right axis) probabilities vary with respect to the packet arrival rate. The optimal policy achieves 37.4\% -- 100.0\% lower outage probabilities (75.3\% when averaged across all data points) and achieves 4.7\% to 100.0\% fewer overflows (47.62\% when averaged across all data points) than the greedy policy.



\begin{figure} [h]
\centering
  \subfloat[Average queue backlog and battery state vs. packet arrival rate]
  {
  \includegraphics[clip, trim = 1cm 9cm 0.5cm 9cm, width=0.9\linewidth]{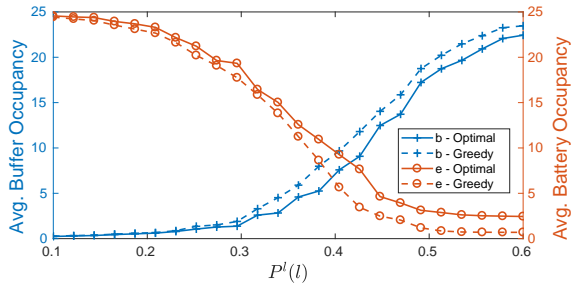}
  \label{fig:avg-buffer-vs-energy-optimal}
  }
  
  \subfloat[Average battery outages and average overflows vs. packet arrival rate]
  {
  \includegraphics[clip, trim = 1cm 9cm 0.5cm 9cm, width=0.9\linewidth]{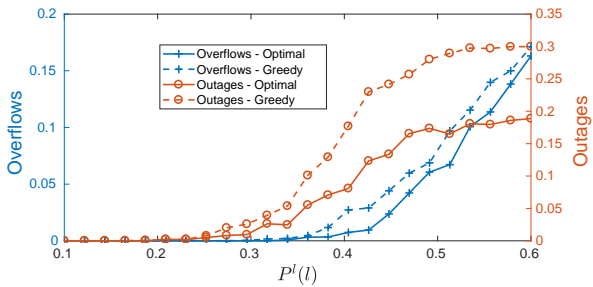}
  \label{fig:avg-outage-vs-avg-overflow}
  }
  
  \caption{Comparison of the optimal and greedy policies.}
  
  \vspace{-0.5cm}
\end{figure}



%% file: conclusion.tex
\section{Conclusion}
\label{sec:con}
We formulated the DSEHS problem as an MDP and analyzed its structural properties. Our analysis does not assume specific data and energy arrival distributions, save for they are i.i.d., and does not require assumptions on the channel transition probabilities, save that they are Markovian. This makes our structural results broadly applicable. We demonstrate that the optimal scheduling policy achieves fewer battery outages, fewer packet overflows, and a better energy-delay trade-off than a greedy policy. As future work, we plan to leverage the structural properties of the DSEHS problem to develop low-complexity reinforcement learning algorithms that can find the optimal scheduling policy with no a priori knowledge of the packet arrival, energy harvesting, and channel dynamics.

%% file: proofs.tex
\appendix


\subsection{Proof of Lemma~\ref{lem:V-non-decr-b}}
The proof follows by induction. Since value iteration converges for any initialization, select $V_0([b, e, h])$ to be non-decreasing in $b$. 
Assume that $V_t([b, e, h])$ is non-decreasing in $b$.
We prove that $V_{t + 1}([b, e, h])$ is also non-decreasing in $b$.
By definition
\begin{align*}
\lefteqn{V_{t + 1}([b, e, h])} \\
&= \min_{a \in \mathcal{A}(b,e)}  \biggl\{  c([b, h], a) + \gamma \mathbb{E}_{b',e',h'}[ V_t([b', e', h']) ]\biggr\}. \\ 
&= \min_{a \in \mathcal{A}(b,e)} Q_{t+1}([b,e,h],a)
\end{align*}
In Lemma~\ref{lem:cost}.\ref{item:incr-cost-b}, we established that the cost function $c([b, h], a)$ is non-decreasing in $b$. Additionally, since $V_t([b, e, h])$ is non-decreasing in $b$ by the induction hypothesis, and $P^{b}(b^\prime | [b, h], a)$ is stochastically increasing in $b$ by Lemma~\ref{lem:stoch-dom-b}, the expected future value is non-decreasing in $b$. It follows that $Q_{t + 1}([b, e, h], a)$ is also non-decreasing in $b$.

Let $a^*$ be the optimal action in state $(b+1, e, h)$. We have
\begin{align*}
V_{t + 1}([b + 1, e, h]) &= Q_{t + 1}([b + 1, e, h], a^*)\\ 
& \geq Q_{t + 1}([b, e, h], a^*) \\ 
&\geq \min_{a \in \mathcal{A}} Q_{t + 1}([b, e, h], a) \\ 
&= V_{t + 1}([b, e, h])
\end{align*}
where the first inequality follows from the fact that $Q_{t + 1}([b,e,h],a)$ is non-decreasing in $b$ and the second inequality follows from optimality. Thus, the optimal value function $V^*$ is non-decreasing in the buffer state $b$.

\subsection{Proof of Lemma~\ref{lem:PDSV-to-V}}
We may express the value function defined in~\eqref{eq:PDSV_to_V} as
\begin{align*}
V(b,e,h) =&~ \min_{a \in \mathcal{A}(b, e)} \left\{b + \mathbb{E}_f [\pds{V}^{*}(b - f, e - a \cdot e_{TX}, h)]\right\}
\\ =&~ b + (1 - a^*) \pds{V}([b,e,h]) + 
\\ &~ a^* q(h) \pds{V}([b,e - e_{TX},h]) + 
\\ &~ a^* (1 - q(h)) \pds{V}([b - 1, e - e_{TX}, h]),
\end{align*}
where $a^* \in \{0,1\}$ is the optimal action in state $(b, e, h)$ and $q(h) \in [0,1]$ is the packet loss rate. If the PDS value function $\pds{V}(\pds{b}, \pds{e}, \pds{h})$ (i) has increasing differences in $\pds{b}$, (ii) has increasing differences in $\pds{e}$, or (iii) has decreasing differences in $(\pds{b}, \pds{e})$, then the results follow from the fact that a non-negative weighted sum of functions with increasing (decreasing) differences has increasing (decreasing) differences.

\subsection{Proof of Proposition~\ref{prop:incr-diff-PDSV-b}}
Consider the value iteration algorithm, which converges for any initial condition. Initialize the PDS value function $\pds{V}_0(\pds{b}, \pds{e}, \pds{h})$ to satisfy~\eqref{eq:incr-diff-PDSV-b}. Assume that~\eqref{eq:incr-diff-PDSV-b} holds for $\pds{V}_t(\pds{b}, \pds{e}, \pds{h})$, for some $t > 0$. We aim to show that~\eqref{eq:incr-diff-PDSV-b} holds for $\pds{V}_{t + 1}(\pds{b}, \pds{e}, \pds{h})$.
Recall from~\eqref{eq:V_to_PDSV} that the PDS value function can be expressed as a function of the conventional value function. The first term on the r.h.s. of~\eqref{eq:V_to_PDSV} has increasing differences in $\pds{b}$. Thus, we only need to show that the the second term on the r.h.s. of~\eqref{eq:V_to_PDSV} has increasing differences in $\pds{b}$. This is implied if the following condition holds:
\begin{multline}\label{eq:cond-V-b}
{V_t}([\pds{b} + l]^{N_b}, e^\prime, h^\prime ) - {V_t}([\pds{b} - 1 + l]^{N_b}, e^\prime, h^\prime ) \\ 
\leq {V_t}([\pds{b} + 1 + l]^{N_b}), e^\prime, h^\prime ) - {V_t}([\pds{b} + l]^{N_b}, e^\prime, h^\prime ),
\end{multline}
where $[x]^N = \min(x,N)$ and $e^\prime = \min (\pds{e} + {e_H}, {N_e})$.
If we let $N_b = \infty$, then~\eqref{eq:cond-V-b} reduces to
\begin{multline*}
{V_t}(\pds{b} + l, e^\prime, h^\prime ) - {V_t}(\pds{b} - 1 + l, e^\prime, h^\prime ) \\
\leq {V_t}(\pds{b} + 1 + l, e^\prime, h^\prime ) - {V_t}(\pds{b} + l, e^\prime, h^\prime ),
\end{multline*}
which holds by Lemma~\ref{lem:PDSV-to-V}.\ref{item:incr-diff-b}. That concludes the proof.


\subsection{Proof of Proposition~\ref{prop:PDSV-submodular-b-e}}
Consider the value iteration algorithm, which converges for any initial condition. Initialize the PDS value function $\pds{V}_0(\pds{b}, \pds{e}, \pds{h})$ to satisfy~\eqref{eq:PDSV-submodular-b-e}. Assume that~\eqref{eq:PDSV-submodular-b-e} holds for $\pds{V}_t(\pds{b}, \pds{e}, \pds{h})$, for some $t > 0$. We aim to show that~\eqref{eq:PDSV-submodular-b-e} holds for $\pds{V}_{t + 1}(\pds{b}, \pds{e}, \pds{h})$.
Recall from~\eqref{eq:V_to_PDSV} that the PDS value function can be expressed as a function of the conventional value function. The first term on the r.h.s. of~\eqref{eq:V_to_PDSV} is submodular in $(\pds{b}, \pds{e})$. Thus, we only need to show that the expected future value (i.e., the second term on the r.h.s. of~\eqref{eq:V_to_PDSV}) is submodular in $(\pds{b}, \pds{e})$. This is implied by the following condition
\begin{multline}\label{eq:cond-V-b-e}
{V_t}([b'' + 1]^{N_b}, [e'' + 1]^{N_e}, h^\prime ) - {V_t}([b'']^{N_b}, [e'' + 1]^{N_e}, h^\prime )
\\ \leq {V_t}([b'' + 1]^{N_b}, [e'']^{N_e}, h^\prime ) - {V_t}([b'']^{N_b}, [e'']^{N_e}, h^\prime ),
\end{multline}
where we use $[x]^N \triangleq \min(x,N)$, $b'' \triangleq \pds{b}+l$ and $e'' \triangleq \pds{e}+e_H$ to keep the equations compact. To verify that~\eqref{eq:cond-V-b-e} holds, we consider the following two cases.

\textbf{Case 1 ($b'' + 1 \leq N_b$):} Assuming that $b'' + 1 \leq N_b$, we may rewrite~\eqref{eq:cond-V-b-e} as follows:
\begin{multline*}
{V_t}(b'' + 1, [e'' + 1]^{N_e}, h^\prime ) - {V_t}(b'', [e'' + 1]^{N_e}, h^\prime ) \\
\leq {V_t}(b'' + 1, [e'']^{N_e}, h^\prime ) - {V_t}(b'', [e'']^{N_e}, h^\prime ).
\end{multline*}
If $e'' + 1 \leq N_e$, then the condition holds by Lemma~\ref{lem:PDSV-to-V}.\ref{item:submod-b-e} and, if $e'' + 1 > N_e$, then both sides are equal; thus, Case 1 holds.

\textbf{Case 2 ($b'' \geq N_b$):} Assuming that $b'' \geq N_b$, we may rewrite~\eqref{eq:cond-V-b-e} as follows:
\begin{multline*}
{V_t}(N_b, [e'' + 1]^{N_e}, h^\prime ) - {V_t}(N_b, [e'' + 1]^{N_e}), h^\prime ) \\
\leq {V_t}(N_b, [e'']^{N_e}, h^\prime ) - {V_t}(N_b, [e'']^{N_e}, h^\prime ),
\end{multline*}
where both sides are equal to 0; thus, Case 2 holds.
This concludes the proof.